\documentclass[11pt]{article}
\usepackage{graphicx}
\usepackage{amsmath,amsthm}
\usepackage{xcolor}
\usepackage{setspace}
\usepackage{amsfonts}
\usepackage{graphicx}
\usepackage{url}
\usepackage{makecell}
\usepackage{fullpage}
\usepackage{mathtools}
\usepackage{subfig}

\newcommand{\mnote}[1]{ \marginpar{\tiny\bf
            \begin{minipage}[t]{0.5in}
              \raggedright #1
           \end{minipage}}}
           
\begin{document}

\title{Lower Bounds for Oblivious Near-Neighbor Search}
\author{Kasper Green Larsen\thanks{{\tt larsen@cs.au.dk}. Aarhus
    University. Work supported by a Villum Young Investigator Grant
    and an AUFF Starting Grant.}
  \and Tal Malkin\thanks{{\tt tal@cs.columbia.edu}. Columbia University.   Work supported in part by the Leona M.~\& Harry B.~Helmsley Charitable Trust.}
  \and Omri Weinstein\thanks{{\tt omri@cs.columbia.edu}. Columbia University. 
Work supported by NSF CAREER Award CCF-1844887.} \and Kevin Yeo\thanks{{\tt kwlyeo@google.com}. Google LLC.}}
\date{}

\newcommand{\kevin}[1]{{\color{green}Kevin: #1}}

\newcommand{\ignore}[1]{}
\DeclarePairedDelimiter{\ceil}{\lceil}{\rceil}
\DeclarePairedDelimiter{\floor}{\lfloor}{\rfloor}

\newtheorem{theorem}{Theorem}
\numberwithin{theorem}{section}
\newtheorem{definition}[theorem]{Definition}
\newtheorem{construct}[theorem]{Construction}
\newtheorem{lemma}[theorem]{Lemma}

\newcommand{\poly}{\text{poly}}
\newcommand{\pr}{\mathsf{Pr}}
\newcommand{\E}{\mathsf{E}}
\newcommand{\negl}{\mathsf{negl}}
\newcommand{\ds}{\mathcal{D}} 
\newcommand{\bP}{\mathbf{X}}
\newcommand{\bU}{\mathbf{U}}
\newcommand{\bT}{\mathbf{T}}
\newcommand{\bq}{\mathbf{q}}
\newcommand{\adv}{\mathcal{A}}
\newcommand{\ann}{\mathsf{ANN}}
\newcommand{\snap}{\mathsf{snapshot}}
\newcommand{\enc}{\mathsf{enc}}
\newcommand{\Enc}{\mathsf{Enc}}

\newcommand{\calB}{\mathcal{B}}
\newcommand{\calP}{\mathcal{P}}
\newcommand{\idx}{\mathsf{idx}}
\newcommand{\supp}{\mathsf{support}}

\newcommand{\op}{\mathsf{op}}
\newcommand{\qop}{\mathsf{query}}
\newcommand{\uop}{\mathsf{insert}}
\newcommand{\init}{\mathsf{preprocess}}

\newcommand{\view}{\mathbb{V}}

\renewcommand{\log}{\lg}
\newcommand{\ANN}{\mathsf{ANN}}
\newcommand{\ONN}{\mathsf{ObvANN}}
\newcommand{\eps}{\epsilon}
\newcommand{\cA}{\mathcal{A}}
\newcommand{\cB}{\mathcal{B}}
\newcommand{\cC}{\mathcal{C}}
\newcommand{\cE}{\mathcal{E}}
\newcommand{\cI}{\mathcal{I}}
\newcommand{\cP}{\mathcal{P}} 
\newcommand{\cG}{\mathcal{G}}
\newcommand{\cM}{\mathcal{M}}
\newcommand{\cN}{\mathcal{N}}
\newcommand{\cD}{\mathcal{D}}
\newcommand{\cX}{\mathcal{X}}
\newcommand{\cQ}{\mathcal{Q}}
\newcommand{\cU}{\mathcal{U}}
\newcommand{\cT}{\mathcal{T}}
\newcommand{\cR}{\mathbf{R}}
\newcommand{\cF}{\mathcal{F}}
\newcommand{\cV}{\mathcal{V}}
\newcommand{\cZ}{\mathcal{Z}}
\newcommand{\cK}{\mathcal{K}}

\newcommand{\st}{\mathsf{st}}
\newcommand{\dy}{\mathsf{dy}}

\newcommand{\addr}{\mathsf{addr}}

\newcommand{\dist}{\Delta}

\maketitle

\begin{abstract}
We prove an  $\Omega(d \lg n/ (\lg\lg n)^2)$ lower bound on the dynamic cell-probe complexity of statistically \emph{oblivious} approximate-near-neighbor search ($\mathsf{ANN}$) over the $d$-dimensional Hamming cube. For the natural setting of $d = \Theta(\log n)$, our result implies an $\tilde{\Omega}(\lg^2 n)$ lower bound, which is a quadratic improvement over the highest (non-oblivious) cell-probe lower bound for $\mathsf{ANN}$. This is the first super-logarithmic \emph{unconditional} lower bound for $\mathsf{ANN}$ against general (non black-box) data structures. We also show that any oblivious \emph{static} data structure for decomposable search problems (like $\mathsf{ANN}$) can be obliviously dynamized with $O(\log n)$ overhead in update and query time, strengthening a classic result of Bentley and Saxe (Algorithmica, 1980).
\end{abstract}

\thispagestyle{empty}
\newpage
\setcounter{page}{1}

\section{Introduction}
\label{sec:intro}

The {\em nearest-neighbor search} problem asks to preprocess 
a dataset $P$ of $n$ input points in some $d$-dimensional metric space, say   
$\Re^d$, so that for any query point $q$ in the space, the data structure can quickly retrieve 
the closest point in $P$ to $q$ (with respect to the underlying distance metric). 
The {\em $r$-near-neighbor} problem is a relaxation of the nearest-neighbor problem, 
which requires, more modestly, 
to return \emph{any} point in the dataset within distance $r$ of the query point $q$ (if any exists). 
The distance parameter $r$ is typically referred to as the {\em radius}.
Efficient algorithms (either offline or online) 
for both the nearest-neighbor and $r$-near-neighbor problems 
are only known for low-dimensional spaces~\cite{Cla88,Mei93}, 
as the only known general solutions for these problems are the naive
ones: either a brute-force search requiring $O(dn)$ time (say, on a
word-RAM), or precomputing the answers which requires prohibitive
space exponential in $d$. 
This phenomenon is commonly referred to as the 
``curse of dimensionality'' in high-dimensional optimization. 
This obstacle is quite 
problematic as nearest-neighbor search primitives are the backbone
of a wide variety of industrial applications as well as algorithm design,
ranging from machine learning~\cite{SDI06} and computer vision~\cite{HS12},
to computational geometry~\cite{CDO02}, spatial databases~\cite{Tya18} and signal
processing~\cite{MPL00} as
some examples.

To circumvent the ``curse of dimensionality'', a further relaxation of the
near(est)-neighbor problem was introduced, resorting to \emph{approximate} solutions,
which is the focal point of this paper. 
In the {\em $(c,r)$-approximate-near-neighbor} problem ($\ANN$),  
the data structure needs to return any point in $P$ that is distance at most $cr$ from the query point $q$, 
assuming that there exists at least one data point in $P$
that is within distance at most $r$ from the query. If all points in the data set $P$
are distance greater than $cr$ from the query point $q$, no point will be reported.
In other words, $(c,r)$-$\ANN$ essentially asks 
to distinguish the two extreme cases where there exists a point in $P$ which is at most $r$-close to the query point $q$, 
or \emph{all} points in $P$ are at $\geq cr$-far from $q$. 
Perhaps surprisingly, the geometric ``gap" in this promise version of the problem turns out to be crucial,   
and indeed evades the ``curse of dimensionality''. A long and influential line of work in geometric 
algorithms based on \emph{locality sensitive hashing} (or, LSH, for short) techniques \cite{IM98, Pan06}
show that the search time for this promise problem (under various $\ell_p$ norms) can be 
dramatically reduced from $\sim n$ to $n^{\delta}$ (for a small constant $\delta$ depending on $r$ and $c$) 
at the cost of a mild space overhead of $n^{1+\eps}$ or even $n \poly\lg n$
in the \emph{static} setting. Interestingly, these upper bounds extend to the more challenging and realistic
\emph{dynamic} setting where points in the dataset arrive online, yielding a dynamic data structure  
with $\poly\lg n$ update time and $n^{\delta}$ query time~\cite{Pan06}.  
For a more detailed exposition of the state-of-the-art on $\ANN$, 
we refer the reader to the following surveys~\cite{AI06, And09}.

On the lower bound side, progress has been much slower. While there has been a considerable 
amount of work on the limits of $\ANN$ in black-box models
of computation
with ``no-coding'' assumptions (e.g., \cite{BV02,KL05}),
the highest \emph{unconditional} lower bound to date is the
$\Omega(d/\log(sw/nd))$ query time lower bound
for any static data structure
by Wang and Yin~\cite{wangyin} as well as Yin~\cite{Yin16}, 
extending previous results of
\cite{PT06,ACP08,PTW08, PTW10}, 
where $s$ denotes the data structure's storage in cells and $w$ is the
word size in bits. 
This is also the highest cell-probe lower bound to date in the \emph{dynamic} setting 
-- the aforementioned bound implies that any (randomized) dynamic data structure for $\ANN$ 
with fast ($\poly\lg n$) update time must have $\tilde{\Omega}(d)$ query time.  
This is in contrast to typical data structure problems, where online lower bounds are known to be higher 
than their static counterparts.   
While this bound is exponentially far from the aforementioned upper bounds, a recurring theme in complexity theory
is that information-theoretic lower bounds are significantly more challenging
compared to black-box bounds, and hence 
lower.
It is widely believed that the logarithmic lower bound is far from tight, especially in the fully dynamic 
setting. Indeed, Panigrahy {\em et al.}~\cite{PTW10} conjecture that the dynamic cell-probe complexity of $\ANN$ 
should be \emph{polynomial},  
but could only prove this for LSH-type data structures (a.k.a ``low contention")
where no single cell is probed too often.
There are also \emph{conditional} (``black-box") lower bounds asserting that
polynomial $\Omega(n^\eps)$ operational time is indeed necessary for the  
{\em offline} version of $\ANN$, under the Strong Exponential-Time Hypothesis
(\cite{ARW17,Wil18,Rub18}). 

\paragraph{Privacy-Preserving Near-Neighbor Search.}
Due to the increasing size of today's datasets,
an orthogonal line of research
has been studied for {\em privacy-preserving near-neighbor search}.
In this scenario, the dataset of points
have been outsourced by a client to a third-party server such
as a cloud storage provider. The client would like to be able
to perform near(est)-neighbor search queries over the outsourced set of
data points. However, the storage of
potentially sensitive data onto an untrusted third-party brings
many privacy concerns. This leads to the natural
problem of whether a client is able to
outsource a data set of points to an untrusted server while
maintaining the ability to perform private near(est)-neighbor queries
over the data set efficiently.

One aspect of privacy is protecting the content of the outsourced
data set. This problem can be addressed by encryption where the
client holds the secret key.
However, the use of encryption does not protect information
leaked by observing the patterns of access to server memory.
Towards that end, the client may wish to implement {\em oblivious access}
where the patterns of access to server memory is independent of both
the content of the data set as well as the queries performed by the client.
In order to focus on the latter problem,
we assume the server's view only contains the patterns of access to server
memory.
Informally, $\delta$-statistical obliviousness implies that
for any two operation sequences of equal length $O_1$
and $O_2$, it must be that $|\view_{\ds}(O_1) - \view_{\ds}(O_2)| \le \delta$
where $\view_{\ds}(O)$ is the distribution of access patterns to server
memory by the data structure $\ds$ executing $O$.
This can later be combined with standard
computational assumptions and cryptographic encryption or
information-theoretic encryption via one-time padding
(if the client can either hold or securely store a random pad) to
ensure privacy of the data set contents.

To address the problem of protecting access patterns, the
{\em oblivious RAM} (ORAM) primitive was introduced by Goldreich
and Ostrovsky~\cite{GO96}.
ORAM considers the scenario where the server holds an array
and the client wishes to either retrieve or update various elements
in the array while guaranteeing oblivious access.
ORAMs are very powerful as they provide a simple transformation
from any data structure into an oblivious data structure.
By executing every access to server memory of any non-oblivious data structure
using an ORAM,
the access pattern of the resulting data structure ends
up being oblivious.
Due to the importance of ORAM, there has been a long line of work
constructing ORAMs.
For example, we refer the reader to some examples:
\cite{PR10,DMN11,GM11,GMOT12,KLO12,SVS13,CLP14,GHL14,BCP16,CLT16,GLO15}.
Recently,
this wave of research
led to the $O(\log n \cdot \log\log n)$ ORAM construction by
Patel {\em et al.}~\cite{PPR18},
and, finally, an $O(\log n)$ ORAM by Asharov {\em et al.}~\cite{AKL18}.
Therefore, we can build an oblivious data structure
with an additional logarithmic overhead compared to the best non-oblivious data structure.

There has also been significant work on the lower bound of ORAMs.
Goldreich and Ostrovksy~\cite{GO96} present an $\Omega(\log n)$ for ORAMs
in the restricted setting of ``balls-and-bins'' model (i.e. a
``non-coding'' assumption) and
statistical security. Larsen and Nielsen~\cite{LN18} extended
the $\Omega(\log n)$ lower bound to the cell-probe model and
computational security matching the aforementioned upper bounds.
Additionally, works by Boyle and Naor~\cite{BN16} as well
as Weiss and Wichs~\cite{WW18} show that any non-trivial
lower bounds for either offline or online,
read-only ORAMs would imply huge breakthroughs
in lower bounds for sorting circuits and/or locally decodable codes.

Going back to the problem of privacy-preserving near-neighbor search,
many works in the past decade
\cite{KS07, MCA07, GKK08, WCK09, PBP10, YLX13, ESJ14, LSP15, WHL16}
attempt to circumvent the additional efficiency overhead incurred by ORAM.
Instead of ensuring oblivious access where the access patterns
are independent of the data set and queries, the access patterns of
many constructions from previous works end up leaking non-trivial
amounts of information.
For example, the access patterns in the constructions by
Wang {\em et al.}~\cite{WHL16} leak
the identity of the point reported by queries.
In more detail, as their work considers the $k$-nearest-neighbor problem,
their algorithms leak the identity of the $k$ encrypted points
that are closest to the query point.
Recent work by Kornaropoulos {\em et al.}~\cite{KPT18}
has shown that this non-trivial
leakage can be abused to accurately retrieve almost all private data.
As a result,
the requirement of oblivious access is integral to ensure privacy
for the near-neighbor problem.
Therefore, several works consider variants of near-neighbor search
with oblivious access such as~\cite{EFG09,SSW09,BBC10,EHK11,SFR18,AHL18,CCD19}
to name a few.


An intriguing question is whether the extra $\Theta(\log n)$
overhead for oblivious data structures over their non-oblivious counterparts
is really necessary. For the problem of RAMs,
it has been shown that the $\Theta(\log n)$ overhead is both necessary
and sufficient~\cite{LN18,PY18}.
Jacob {\em et al.}~\cite{JLN18} also show that
the $\Theta(\log n)$ overhead is necessary and sufficient for
many fundamental data structures such as stacks and queues, but quite
surprisingly, Jafargholi {\em et al.}~\cite{jafargholi} very recently showed that
(comparison-based) priority queues can be made oblivious with no overhead at all.
We consider this question for the $\ANN$ problem. 
In particular, is it possible to prove a logarithmically larger
lower bound for the oblivious $\ANN$ problem
as opposed to the best known non-oblivious $\ANN$ lower bound?
We answer in the affirmative in this work.


\subsection{Our Contributions}
Our main result is a stronger cell-probe lower bound
for the oblivious $\ANN$ problem,
which is $\tilde{\Omega}(\log n)$ higher than the best known cell-probe
lower bound for the non-oblivious $\ANN$ problem.

\begin{theorem}[Informal]  \label{thm_main}
Let $\ds$ be any dynamic, statistically oblivious data structure that solves
$(c,r)$-$\ANN_{d,\ell_1}$ over the $d$-dimensional Hamming cube, 
on an online sequence of $n$ insertions and queries, in the oblivious cell-probe model 
with word size $w$ and client storage of $m = o(n)$ bits.    
Then for some constant $c > 1$ and $r = \Theta(d)$, 
$\ds$ must have worst case per-operation running time
$$\Omega\left( \frac{d \cdot \log(n/m)}{(\log(w \log n))^2} \right).$$
In the natural setting  of $m \le n^{1-\rho}$ and $w = \Theta(\log n)$, the operational time is at least $\Omega(d\log n/(\log\log n)^2)$.
\end{theorem}

To the best of our knowledge, this is the first time that
a lower bound of $\omega(d)$ has been successfully proved
for $\ANN$ in the cell-probe model.
This is also the first oblivious cell-probe lower bound exceeding $\omega(\log n)$.
Previous works on oblivious cell-probe lower bounds have focused
on data structures with $O(\sqrt{\log n})$ or smaller complexity for their non-oblivious
counterparts (such as RAMs~\cite{LN18,PY18} as well as stacks, queues, deques,
priority queues and search trees~\cite{JLN18}) and peaked at $\Omega(\log n)$.
On the technical side, we remark that our work is the first to apply the technique
of Larsen~\cite{Lar12} of
combining the chronogram~\cite{FS89} with cell sampling~\cite{PTW10}
to prove a lower bound on privacy-preserving data structures. 
So far, these techniques could not be leveraged to prove higher bounds 
in the oblivious cell-probe model.


To complement our main result, 
we present a variant of the reduction by Bentley and Saxe~\cite{BS80}, 
who showed that dynamic data structures can be built in a black-box fashion 
from their static counterparts, for the special class 
of {\em decomposable} problems (which include many natural variants of
near-neighbors search, range searching and any class of linear queries).
We show that
any oblivious static data structure solving a decomposable problem
can be transformed
into an oblivious dynamic data structure with only an additional
logarithmic overhead.

\begin{theorem}[Informal]
If there exists an oblivious static data structure for a decomposable problem $\calP$ of $n$ items
with storage of $S^\st(n)$ cells, preprocessing of $P^\st(n)$ cell-probes and
amortized $Q^\st(n)$ cell
probes for queries,
then there exists an oblivious dynamic data structure for $\calP$ using
$S^\dy(n) = O(S^\st(n))$ cells of storage, preprocessing of
$P^\dy(n) = P^\st(n)$ cell probes,
	amortized $Q^\dy(n) = O(\log n \cdot Q^\st(n) + \log n \cdot P^\st(n)/n)$ cell probes for each query/update operation.
\end{theorem}

The above theorem states that the largest separation between
oblivious cell-probe lower bounds for static and dynamic structures
solving decomposable problems can be at most logarithmic.
One can view the chronogram technique as creating a dynamic data structure
lower bound by boosting a static data structure lower bound (via
the cell sampling method) by an $\tilde{\Omega}(\log n)$ factor.
Therefore, the chronogram can be viewed as optimal for decomposable problems
even in the oblivious model.

\subsection{Technical Overview}     

The high-level approach behind the proof of Theorem \ref{thm_main} 
is to exploit obliviousness in a new (and subtle) way in order to compose a variation of
the \emph{static} cell-sampling lower bound for $\ANN$ in~\cite{PTW10} together
with the chronogram method~\cite{FS89}.
While this template was the technical approach of
several previous dynamic data structure lower bounds for queries 
with ``error-correcting codes" (ECC) properties     
(such as polynomial evaluation~\cite{Lar12}, range counting~\cite{Lar12b} and
online matrix-multiplication~\cite{CGL15}), 
this program is doomed to fail for $\ANN$ for two fundamental reasons.
The first reason is that the chronogram method requires the underlying data structure problem  
to have an ``ECC-like" property, namely, that any \emph{local} modification of the database 
changes the answer to (say) \emph{half} of the queries
(in other words, a random query is {\em sensitive}  
to even a single update in the data set).
In contrast, $\ANN$ queries are sensitive 
only to updates in an exponentially-small volumed ball around the query point. 
This already impedes the application of the chronogram method.  
The second, more subtle and technically challenging problem, is the fact that in the $\ANN$ problem,  
only a tiny fraction ($1/\poly(n)$) of queries
actually reveal information 
about the underlying data set -- these are queries which reside \emph{close} to the data set and hence 
may report an input point (we call these ``yes" queries). 
As explained below, this feature of $\ANN$ turns out to be a significant barrier in
carrying over the static cell-sampling argument to the \emph{dynamic} setting (as opposed 
to cell-sampling lower bounds for ``$k$-wise independent" queries),  
and overcoming this problem is the heart of the paper. 
Surpassing this obstacle also entailed us to construct an alternative 
information-theoretic proof of \cite{PTW10}'s \emph{static} lower bound for the standard
(non-oblivious) $\ANN$ problem, which is key for scaling it to the dynamic setting (and,
as a bonus, also improves the parameters of the lower bound in~\cite{PTW10}).

In order to overcome the aforementioned two challenges, we use obliviousness in \emph{two different} ways.  
The first one, which is more standard
(in light of recent works~\cite{LN18,PY18}), 
overcomes the first problem, mentioned above, of {\em insensitivity}
of near-neighbor queries to the chronogram construction.
Recall that the chronogram method partitions a sequence of $\Theta(n)$ random update
operations into $\tilde{\Theta}(\lg n)$ geometrically decreasing intervals 
(``{\em epochs}"), where the hope is to show that a random query is 
\emph{simulteneously} sensitive to (essentially) all epochs. 
As discussed above, $\ANN$ lacks this property, and it is not hard to see that 
if, for example, updates are drawn uniformly and independently at random,
then any query will only be sensitive to the first $O(1)$ epochs 
with overwhelming probability (due to the geometric decay of epochs, 
which is essential, as it reduces a dynamic lower bound to that of solving logarithmically many 
\emph{independent} static problems, one per epoch). 
We circumvent this issue by using a simple geometric partitioning trick of the hypercube together 
with the (computational) indistinguishability constraint of ORAMs. This argument is key to the proof, 
as it reverses the quantifiers: it implies that for oblivious data structures, it is enough to show that 
for each epoch, there is \emph{some} distribution on $\ANN$ queries that must read 
$\tilde{\Omega}(d)$ cells from the epoch (as opposed to a \emph{single} distribution which 
is sensitive to \emph{all} epochs). Indeed, assuming this (much) weaker condition, if the data structure does not probe 
$\tilde{\Omega}(d)$ cells from \emph{every} epoch, an adversary (even when computationally bounded)
can learn information about the query's location (in particular, which partition the query belongs to), contradicting obliviousness.



The second way in which we exploit obliviousness is much more subtle
and illuminates
the difficulty in carrying out cell-sampling static
lower bounds in dynamic settings for
data structure problems (like $\ANN$) where only $o(1)$-fraction of the queries
reveal useful information.
Before diving into the dynamic case,  we briefly explain our modifications
of the static lower bound which enables 
a higher lower bound in the dynamic setting.
At a high level, the cell-sampling argument of~\cite{PTW10}
shows that for very efficient, static data structures, there exists
a small number of {\em memory cells} $T$ of the data structure
that are the only cells
probed by many queries.
These queries are referred to as {\em resolved queries}
and denoted by $Q(T)$.
The main idea of cell sampling is to show that the queries in $Q(T)$ reveal
more bits of information about the underlying data set (denoted by $\bP$) than
the number of bits that can be stored in the sampled cells $T$, which would
lead to a contradiction.
However, in the $\ANN$ setting, showing that the queries in $Q(T)$
reveal enough information about the underlying data set $\bP$
is highly nontrivial -- 
One way to prove this statement is to show that
the resolved queries are essentially \emph{independent} of the
underlying data set, i.e., $Q(T) \perp \bP$.
If this were true, then a standard \emph{metric expansion} argument shows that the
{\em neighborhood} of distance $r$ surrounding all resolved queries $Q(T)$,
covers at least \emph{half} of the boolean hypercube.
As a result, all points landing in the neighborhood of $Q(T)$ will be reported
by at least one query in $Q(T)$.
If the points in the data set are generated
\emph{uniformly and independently} distributed conditioned on $Q(T)$,
it can be shown that a constant fraction of data set points in $\bP$
will fall into neighborhood of $Q(T)$ except with negligible probability.
Hence, a constant fraction of points in $\bP$ will be recovered by using only
the contents of sampled cells $T$.
Alas, for \emph{adaptive} data structures,
the resolved queries could depend heavily on the content of the cells,
and this \emph{correlates} $Q(T)$ and the database $\bP$.
In the work of~\cite{PTW10}, the authors handle this correlation using a
careful,
adaptive cell-sampling argument combined
with a union-bound over all possible memory states of the data structure,
which effectively breaks the dependence between resolved queries $Q(T)$
and the data set $\bP$.
We present an alternative method of proving independence using
information theoretic arguments.
Intuitively, even though $Q(T)$ and $\bP$ are indeed correlated random variables  
in the general adaptive setting, we argue that this correlation cannot be too large: 
the set of resolved queries $Q(T)$ are completely
determined by the addresses and contents of the sampled cells $T$, 
as one can determine whether $q \in Q(T)$
by executing $q$ and checking if $q$ ever probes a cell outside of $T$.
Since $T$ is a small set of cells,
the data set $\bP$ and the set of
resolved queries $Q(T)$ have low mutual information by a data processing inequality.
We formalize this intuition by constructing an impossible 
``geometric packing" compression argument of the data set $\bP$ using only the
sampled cells $T$.
These ideas also allow us to use {\em one-round} cell sampling~\cite{Lar12b}
as opposed to multiple-round cell-sampling, which slightly improves
the lower bound shown in~\cite{PTW10}.

Moving back to the dynamic setting, our new arguments still break down
due to the fact that memory cells may be overwritten at different points in time.
The typical method for proving dynamic lower bounds~\cite{Lar12,Lar12b}
composes the cell sampling technique and chronogram method.
A random update sequence $\bU$ is partitioned into geometrically-decreasing
sized epochs. For epoch $i$,
we denote $C_i(\bU)$ as all cells that were last overwritten by updates in epoch $i$, $\bU_i$.
Next, the cell sampling technique is applied to each $C_i(\bU)$ to
find a small subset of sampled cells $T_i \subseteq C_i(\bU)$ such that
for almost all queries, the only cells probed in $C_i(\bU)$ appear in $T_i$.
We denote these resolved queries by $Q_i(T_i)$.
Once again, we need to show that the answers of resolved queries $Q_i(T_i)$
reveal a lot of information about points inserted in $\bU_i$.
Unfortunately, our previous approach fails as
it is \emph{impossible} to determine $Q_i(T_i)$ using
only the sampled cells $T_i$.
Note, if a query probes a cell outside of $T_i$, one cannot determine
whether the cell belongs to $C_i(\bU)$ or not.
Therefore, one needs to know the addresses
of cells in $C_i(\bU)$, denoted by $C^{\addr}_i(\bU)$,
to determine $Q_i(T_i)$.
Unfortunately,
the number of bits needed
to express $C^{\addr}_i(\bU)$ may be very large and contain
significant information about $\bU_i$.
So, we can no longer argue that the set of resolved queries $Q_i(T_i)$ is
determined by a low-entropy random variable as in the static case.


This is where \emph{statistical obliviousness} comes to the rescue.
The main observation is that the 
\emph{addresses} of cells last overwritten by updates in epoch $i$, $C^{\addr}_i(\bU)$,
cannot reveal too much information about the updates in $\bU_i$ for any  
sufficiently statistically oblivious data structure.
We prove this using a certain ``reverse Pinsker inequality" which allows 
us to conclude that the mutual information $I(\bU_i;C^{\addr}_i(\bU)) = o(|\bU_i|)$ bits
for any $O(1/\log^2 n)$-statistically oblivious data structure. We note this inequality
may be of independent interest 
to other oblivious lower bounds.
Now, we can see that the address sequence $C^{\addr}_i(\bU)$, together with the \emph{small} set
of sampled cells $T_i \subseteq C_i(\bU)$ from the $i$-th epoch,
\emph{completely determine} the resolved query set $Q_i(T_i)$.
Therefore, a data processing argument once again asserts that the
large resolved query set $Q_i(T_i)$ is
\emph{almost independent} of the updates $\bU_i$.
By a packing argument (similar to the static case), we can show
that a constant fraction of the points in $\bU_i$ fall into the neighborhood around
the resolved queries
$Q_i(T_i)$ and each of these points
will be returned by at least one resolved query.
As a result, the answers of resolved queries reveal more information
about $\bU_i$ than the number of bits that can be stored in the sampled cells $T_i$
providing our desired contradiction.
We conclude that at least $\tilde{\Omega}(d)$ cells must be probed from each epoch.
Combined with our first application of obliviousness, we show that
$\tilde{\Omega}(d\log n)$ cells must be probed from all epochs.

\subsection{Related Work} 

The cell-probe model was introduced by Yao~\cite{Yao81} 
as the most abstract (and compelling) model for proving lower bounds on the operational time of data structures, 
as it is agnostic to implementation or hardware details, and hence captures any imaginable 
data structure. The \emph{chronogram technique} of Fredman and Saks~\cite{FS89} was the first to prove
$\Omega(\log n / \log\log n)$ dynamic cell-probe lower bounds.
P\v{a}tra\c{s}cu and Demaine~\cite{PD06} later introduced the
information transfer technique which was able to prove $\Omega(\log n)$
lower bounds. Larsen~\cite{Lar12} was able to combine the chronogram
with the cell-sampling technique of static data structures~\cite{PTW10}
to prove an $\Omega((\log n / \log\log n)^2)$ for \emph{range searching} problems, which remains
the highest cell-probe lower bound to date for any dynamic search problem.
Recently, Larsen {\em et al.}~\cite{LWY18} exhibited a new technique 
for proving $\tilde{\Omega}(\log^{1.5} n)$ cell-probe lower bounds on \emph{decision} 
data structure problems, circumventing the need for large outputs (answer length) in previous lower bounds.

\paragraph{Oblivious cell-probe lower bounds.}{
The seminal work of Larsen and Nielsen~\cite{LN18} presented
the first cell-probe lower bound for \emph{oblivious data structures}, in which they proved 
a (tight) $\Omega(\log n)$ lower bound for ORAMs.
Jacob {\em et al.}~\cite{JLN18} show $\Omega(\log n)$ cell-probe lower
bounds for oblivious stacks, queues, deques, priority queues
and search trees.
Both~\cite{LN18,JLN18} adapt the information transfer technique of
P\v{a}tra\c{s}cu and Demaine~\cite{PD06}.
Persiano and Yeo~\cite{PY18}
show an $\Omega(\log n)$ lower bound for differentially private RAMs
which have weaker security notions than ORAMs
using the chronogram technique originally introduced by
Fredman and Saks~\cite{FS89} with modifications by
P\v{a}tra\c{s}cu~\cite{Pat08}.
Another line of work has investigated the hardness of lower bounds
for other variants of ORAMs.
Boyle and Naor~\cite{BN16} show that lower bounds
for offline ORAMs (where all operations are given in batch before
execution) imply lower bounds for sorting circuits.
Weiss and Wichs~\cite{WW18} show that lower bounds for
online, read-only ORAMs imply lower bounds for either
sorting circuits and/or locally decodable codes.}

\paragraph{Near-neighbor lower bounds.}
There have been many previous works on lower bounds for non-oblivious
near(est)-neighbors problems. The following series
of lower bound results considered deterministic algorithms
in polynomial space~\cite{BOR99,BR02,CCG03,Liu04}.
Chakrabarti and Regev~\cite{CR04} present tight lower bounds
for the approximate-\emph{nearest}-neighbor problem for possibly randomized
algorithms that use polynomial space.
Several later works consider various lower bounds for near(est)-neighbors 
with different space requirements, the ability to use randomness
and different metric spaces~\cite{CR04,PT06,AIP06,ACP08,PTW08}.
As mentioned before, the highest
cell-probe lower bound for dynamic $\ANN$ is the static 
$\Omega(d/\log(sw/dn))$ lower bound of Wang and Yin~\cite{wangyin}. In fact,  
all the above works prove lower bounds on static near-neighbor search
where the data set is fixed and no points may be added.

\section{Preliminaries}
We present a formal definition
of the oblivious cell-probe model as well as the
$\ANN$ problem.

\subsection{Oblivious Cell Probe Model}

We will prove our lower bounds in the {\em oblivious cell-probe} model
which was introduced by Larsen and Nielsen~\cite{LN18} and is an
extension of the original cell-probe introduced by Yao~\cite{Yao81}.
The oblivious cell-probe model consists of two parties:
the {\em client} and the {\em server}.
The client outsources the storage of data to the adversarial
server which is considered to be honest-but-curious
(also referred to as semi-honest).
In addition, the client wishes to perform some set of operations
over the outsourced data in an {\em oblivious} manner.
Oblivousness refers to the the client's wishes
to hide the operations performed on the data
from the adversarial server that views
the sequence of cells probed in the server's memory.
Note the adversary's view does not contain the contents of server
memory as a way to separate the security of accessing data and
securing the contents of data.
We now describe the oblivious cell-probe model in detail.

In the oblivious cell-probe model, the server's memory consists of cells
with $w$ bits. Each cell is given a unique address from the set of
integers $[K]$. It is assumed that all cell addresses can fit into a single
word which means that $w \ge \ceil{\log_2 K}$.
The client's memory consists of $m$ bits.
Additionally, there exists an arbitrarily long, finite length
binary string $\cR$ which contains all the randomness
that will be used by the data structure.
For cryptographic purposes, $\cR$ may also be used as a random oracle.
The binary string $\cR$ is chosen uniformly at random before
the data structure starts processing any operations. As a result,
$\cR$ is independent of any operations of the data structure.

A data structure in the oblivious cell-probe model performs
operations that only involve either a {\em cell probe} to server memory
or accessing bits on client memory. During a cell probe in server memory,
the data structure is able to read or overwrite the contents of the probed
cell.
The cost of any operation is measured by the number of cells that
are probed on the server's memory.
The accesses to bits in client memory are considered free for the
data structure. Any access to bits in the random string $\cR$
are also free.
We denote the {\em expected query cost} to be the maximum
over all sequences of operations $O$ and query $q$ of
the expected number of cell probes performed when answering query $q$ over the
random string $\cR$
after processing the all operations in $O$. We denote the {\em worst
case update cost} as the maximum over all sequences of
operations $O$, update $u$ and random string $\cR$
of the number of cells probed when processing update $u$ after
processing all operations in $O$.

We now move onto the privacy requirements of data structures
in the oblivious cell-probe model. The random variable
$\view_{\ds}(Q)$ as the {\em adversary's view} of the data structure
$\ds$ processing a sequence of operations where randomness
is over the choice of the random string $\cR$.
The adversary's view, $\view_{\ds}(O)$, will contain the addresses of cells that are
probed by $\ds$ when processing
$O$.
Finally, we assume that $\ds$ must process a sequence of
operations in an online manner. That is, $\ds$ must finish executing
one operation before receiving the next operation.
Furthermore, the adversary is aware when execution of one operation
finishes and the execution of another operation begins.
As a result, for any sequence $O = (\op_1,\ldots,\op_n)$,
we can decompose the adversary's view as $\view_{\ds}(O) = (\view_{\ds}(\op_1), \ldots, \view_{\ds}(\op_n))$.
Unlike the previous model, we will assume statistical security
instead of computational security.
We now present a formal definition of the security of an oblivious cell-probe
data structure.

\begin{definition}
A cell-probe data structure $\ds$
is $\delta$-statistically oblivious if for any two equal length sequences
$O_1$ and $O_2$ consisting of
valid operations, then the statistical distances of $\view_\ds(O_1)$ and
$\view_\ds(O_2)$ satisfy
$$|\view_\ds(O_1) - \view_\ds(O_2)| \le \delta.$$
\end{definition}

Throughout the rest of our work, we will consider $\delta$-statistical
obliviousness with $\delta \le 1/\log^{2} n$.
Note that the above definition is a much weaker definition
than previous definitions of statistical obliviousness in cryptography
as the distinguishing probability need be at most $1/\log^2 n$ as opposed
to being a negligible function of $n$. However, as we
are proving a lower bound, a weaker notion of obliviousness results in
strong lower bounds.

We now briefly describe the implications of cell-probe lower bounds
in the client-server setting. The majority of previous ORAM works
considered the server to be {\em passive storage}, which means that the server
does not perform any computation beyond retrieving and overwriting the contents of cell
at the request of the client. In this case, a cell-probe lower bound implies a bandwidth
lower bound in the client-server setting for any oblivious data structure.
On the other hand, if we consider the case when the server can perform arbitrary computation,
any cell-probe lower bound implies a lower bound on server computation.




\subsection{Approximate-Near-Neighbor ($\ANN$) Problem}

We now formally define the $(c,r)$-approximate-near-neighbor problem
over the $d$-dimensional boolean hypercube using the $\ell_1$ distance as the measure.
In our work, we focus on the online version which allows insertion of points into the dataset.
Let $U := \{0,1\}^d$ denote the set of all points in the space.
If the $\uop$ operation is called with the same point $p \in U$ twice,
then the second $\uop$ operation is ignored. We now
formally describe the problem.

\begin{definition}[Online, Dynamic $(c,r)$-$\ANN_{d, \ell_1}$]
The dynamic $(c,r)$-approximate-near-neighbor problem over the
$d$-dimensional boolean hypercube endowed with the $\ell_1$ distance measure
asks to design a data structure
that maintains an online dataset $S\subset U$
under an online sequence of $n$ operations of the following two types:   
\begin{enumerate}
	\item $\uop(p), p \in U$: Insert the point $p$ if it does not already exist in $S$;
	\item $\qop(q), q \in U$: If there exists a unique
		$p \in S$ such that $\ell_1(p, q) \le r$, then report any
		any $p' \in S$ such that $\ell_1(p',q) \le cr$.
		If all points $p \in S$ are such that $\ell_1(p, q) > r$, then the output should be $\perp$.
\end{enumerate}
\end{definition}

\subsection{Decomposable Problems}
We now define decomposable problems. From a high level, a problem is
decomposable if the problem may be solved on partitions of any
data set and the results can be combined to give the result over
the entire data set. Many natural problems are decomposable such
as many variants of near-neighbors search, range counting and
interval stabbing.

\begin{definition}
A problem $\calP$ is {\em decomposable} if for any two disjoint data sets
$D_1$ and $D_2$ and any query $q$,
there exists a function $f$ that can be computed
in $O(1)$ time such that
$$
\calP(q, D_1 \cup D_2) = f(\calP(q, D_1), \calP(q, D_2)).
$$
\end{definition}

\section{Oblivious, Dynamic Lower Bound}

In this section, we prove a logarithmically larger lower bound for the
dynamic variant of the $\ANN$ problem compared to the previous, highest lower bound
for non-oblivious $\ANN$ by Wang and Yin~\cite{wangyin}.
We consider the $(c,r)$-$\ANN_{d, \ell_1}$ problem over a
$d$-dimensional boolean
hypercube with respect to the $\ell_1$ norm
where $d = \Omega(\log n)$ and $\Theta(n)$ will be number of points inserted
into the data set.
We denote $t_u$ as the worst case time for any $\uop$ operation
and $t_q$ as the expected time for any $\qop$ operation. 
For our oblivious lower bound, we consider the two party scenario
where a client stores $m$ bits that are free to access while the server
holds the cells consisting of the data structure's storage.
We prove the following lower bound:

\begin{theorem}
\label{thm:main}
Let $\ds$ be an randomized, dynamic, oblivious cell-probe data structure for $(c,r)$-$\ANN_{d, \ell_1}$
over a $d$-dimensional boolean hypercube where $d = \Omega(\log n)$
under the $\ell_1$ norm.
Let $w$ denote the cell size in bits, $S$
denote the number of cells stored by the server for the data structure and
$m$ denote the client storage in bits. If $m = o(n)$,
then there exists parameters of constant
$c \ge 1$ and $r = \Theta(d)$
and a sequence of $\Theta(n)$ 
operations such that
$$t_q = \Omega\left( \frac{d \cdot \log(n/m) }{(\log(t_u w))^2}\right).$$
\end{theorem}

To prove Theorem \ref{thm:main}, we proceed with a ``geometric variation"  
of the chronogram argument in~\cite{Lar12} where our operation sequence consists of $\Theta(n)$ independent 
\emph{but not identically drawn} random $\uop$ operations, which we will describe later.
This random sequence of updates is followed by a single $\qop$ operation.
The $\uop$ operations
are partitioned into {\em epochs} whose sizes decrease exponentially by a parameter
$\beta \ge 2$ which will be defined later. All epochs will
contain at least $\max\{\sqrt{n}, m^2\}$ $\uop$ operations.
Each epoch will be indexed by an non-negative integer that increases in reverse chronological time.
Epoch 0 will consist of the last $\max\{\sqrt{n}, m^2\}$ $\uop$ operations before the query is performed, epoch
1 will consist of the last $\beta \cdot \max\{\sqrt{n}, m^2\}$ $\uop$ operations before epoch 1 and so forth.
Therefore, there will be $k := \Theta(\log_\beta (n/m))$ epochs.
For all epochs $i$ where $0 \le i < k$ will consist of exactly $n_i := \beta^i \cdot \max\{\sqrt{n}, m^2\}$ $\uop$ operations.

For notation, the sequence of $n$ $\uop$ operations are denoted by the random variable $\bU$.
We denote
the sequence of $\uop$ operations in any epoch indexed by $i$ using the random
variable $\bU_i$. Therefore, we can write $\bU = (\bU_{k-1}, \ldots, \bU_0)$.

\paragraph{Hard distribution.}
Given the partitioning of the $\Theta(n)$ $\uop$ operations into geometrically decaying epochs, 
we now define the hard distribution for our lower bound. 
In order to (later) exploit the obliviousness of the data structure, our 
hard distribution shall have a ``direct sum" structure, which is simple to 
design using the geometry of the $\ANN$ problem. 
Conceptually, the hard distribution will split the $d$-dimensional
boolean hypercube into \emph{disjoint subcubes} where each 
subcube is uniquely assigned to one of the epochs.
To this end, every epoch $i \in \{0,\ldots,k-1\}$ will be assigned a $d'$-dimensional
boolean subcube where $d' := \Theta(d)$ will
be determined later.
Each of the $\uop$ operations of any epoch $i$ will be generated independently
by picking a point from epoch $i$'s $d'$-dimensional boolean hypercube uniformly at random.

We now show how we split up the original $d$-dimensional boolean
hypercube into $k$ $d'$-dimensional boolean subcubes that
are disjoint. We choose the parameter $d > d'$
where $d'$ will be specified later. We assign each of the $k$ epochs
a unique prefix of $d - d'$ bits denoted by
$p_0,\ldots,p_{k-1} \in \{0,1\}^{d - d'}$ where $p_i$
is the prefix for epoch $i$. We will pick the prefixes in such a way
that for any $i \ne j \in [k]$, $\ell_1(p_i, p_j) > d'$.
To see that such a choice of prefixes exist, we consider the following
probabilistic method where we pick the $k$ prefixes of $d - d'$ bits
uniformly at random. For any two $i \ne j \in [k]$ and sufficiently
large $d = \Omega(\log n)$, we know that
$$
\Pr[\ell_1(p_i, p_j) < 0.49(d - d')] \le 1/n^3.
$$
By a Union bound over all $n^2$ possible pairs, we get
that there must exist some choice of $k$ prefixes
such that pairwise prefixes have $\ell_1$ distance at least $d'$
as long as $d \ge 4d'$. The $d'$-dimensional subcube for epoch $i$
is constructed as all points in the original $d$-dimensional subcube
restricted to the case that the $d-d'$ coordinates match the
prefix $p_i$.
We note that our choice of subcubes has the important property
that two points from different subcubes will be distance
at least $d'$ from each other as their prefixes already have
$\ell_1$ distance of at least $d'$.

Before continuing, we describe why this choice of hard distribution
is compatible with oblivious data structures.
Intuitively, our choice of hard distribution is very revealing for the choice of update
points.
An adversary is aware that updates from epoch $i$ will be completely contained
in the subcube assigned to epoch $i$. Furthermore, as all subcubes are pairwise disjoint,
two update points from different epochs cannot be from the same subcube. We will exploit
this fact in combination with the oblivious guarantees to prove lower bounds on the operational
cost of the final query. If, on average, a query point
does not probe many cells that were last overwritten
in some epoch $i$, then the adversary can simply rule out that the query point was chosen
from the disjoint subcube assigned to epoch $i$. This knowledge learned by the adversary
can be used to contradict the obliviousness property. As a result, we can show
that an oblivious data structure must query many cells last written from all epochs
to hide the identity of the query point even if the query needs no information
from some epochs.

Formally, we define the distribution of updates in epoch $i \in \{0,\ldots,k-1\}$, $\bU_i$,
as the product of $n_i$ identical distributions, $\mu_i$.
The distribution $\mu_i$ deterministically appends the prefix $p_i$ uniquely
assigned to epoch $i$ and picks the remaining $d'$ coordinates uniformly at
random.
We denote this $d'$-dimensional subcube using $P_i$.
The entire distribution of updates over all epochs, $\bU$, can be viewed as the product
of distribution $\bU = \bU_{k-1} \times \ldots \times \bU_0$.
Our hard query distribution $\bq$ will simply be to query any fixed point
that lies outside each of the subcubes $P_0,\ldots,P_{k-1}$.

We show that the probability that any two points inserted
during $\bU_i$ are too close is low.

\begin{lemma}
\label{lem:dis}
Let $\bU_i$ be the set of update points inserted in epoch $i$
according to the hard distribution. For sufficiently
large $d' = \Omega(\log n)$, there cannot
exist any query $q$ such that $\ell_1(u, q) \le 0.24d'$ and $\ell_1(v, q) \le 0.24d'$
for any two different points $u$ and $v$ chosen by $\bU_i$ except with probability
at most $1/n$.
\end{lemma}
\begin{proof}
Note both $u$ and $v$ are chosen uniformly at random from a
$d'$-dimensional boolean hypercube. As a result, we know that $\E[\ell_1(u, v)] = 0.5d'$.
We apply Chernoff Bounds over the coordinates of $u$ and $v$ to get
that $\Pr[\ell_1(u, v) \ge 0.49d'] \le 1/n^3$ for sufficiently large
$d' = \Theta(d) = \Omega(\log n)$. Next, we apply a Union Bound over all
${n \choose 2} \le n^2$ pairs of points in $\bP$. As a result,
the probability of the existence of two points $u$ and $v$
whose distance is at most $0.49d'$ is at most $1/n$.

Suppose there exists a query $q$ such that $\ell_1(u, q) \le 0.24d'$
and $\ell_1(v, q) \le 0.24d'$. By the triangle inequality, we know that
$\ell_1(u, v) \le 0.48d' < 0.49d'$. This only occurs with probability at most
$1/n$.
\end{proof}

Additionally, we also want that queries cover large portions of the boolean
hypercube cube such that they must report a point if it lands
in these large subspaces of the boolean hypercube.
We quantify this by considering the {\em neighborhood} of subsets of queries
over the boolean hypercube. For any query $q$, we consider
its neighborhood to be all points in the boolean hypercube that are
distance at most $r$ from $q$. For subsets of queries within
an epoch's assigned subcube denoted by $Q \subseteq \{0,1\}^{d'}$,
we consider the neighborhood of $Q$ to be any points within distance
$r$ of any query $q \in Q$. We denote the neighborhood of $Q$
by $\Gamma_r(Q)$. We will use the following
standard isoperimetric inequality describing the size of neighborhoods
over any $d'$-dimensional boolean hypercube which follows directly from
Harper's theorem~\cite{FF81}.

\begin{lemma}
\label{lem:exp}
Let $H$ be all the vertices of a $d$-dimensional boolean hypercube.
Let $V$ be a subset of vertices in $H$ such that
$|V| \le 1/(2a^{\epsilon^2 d}) \cdot |H|$ and let $\Gamma_{\epsilon d}(V)$ be the set of all vertices that are distance
at most $\epsilon d$ from any of the vertices in $V$. Then, there exists some
constant $a > 1$ such that
$$
|\Gamma_{\epsilon d}(V) | \ge a^{\epsilon^2 d} \cdot |V|.
$$
\end{lemma}

For convenience, we denote $\Phi := \Phi(r) := a^{\epsilon^2 d'}$ as the {\em expansion}
over each of the $d'$-dimensional boolean hypercubes for distances
of $r := \epsilon \cdot d'$ where $0 < \epsilon < 1$ is a constant.
The above lemmata will end up being important later when we prove our lower bounds.


\paragraph{Choosing parameters.}
We now choose the parameters for our problem. First, we want
to ensure that if a query $q$ may report any point, that point will be unique
with high probability. We can ensure this property by picking $cr \le 0.24d'$ and 
applying Lemma~\ref{lem:dis}. To ensure large expansion within each epoch's subcube,
we will set $r = \Theta(d')$. As an example, we can choose parameters such as
$r = 0.01d'$ and $1 \le c \le 24$ to get our desired properties.



\subsection{Overview of Our Proof}

Before we begin formally proving our lower bound, we present a high level overview showing
the steps of our approach. Our techniques will follow the techniques first outlined by
Larsen~\cite{Lar12}, which combine the chronogram introduced by
Fredman and Saks~\cite{FS89} and the cell sampling method introduced
by Panigrahy {\em et al.}~\cite{PTW10}.
We fix $t_u$ to be the worst case update time and our goal is to prove a lower bound on the expected query time $t_q$.

For the sequence of $\Theta(n)$ randomly chosen
$\uop$ operations $\bU$, we denote $C(\bU)$ as the random variable of the set of
all the cells stored by the data structure
after processing all $\uop$ operations of $\bU$.
We partition the cells of $C(\bU)$ into $k$ groups depending on the most recent
operation that updated the contents of the cell.
In particular, we denote $C_i(\bU)$ as the random variable
describing the set of cells in $C(\bU)$ whose contents were last updated by an $\uop$ operation performed during epoch $i$. For any query point $q \in Q$, we denote $t_i(\bU, q)$ as the
random variable denoting the number of cells that are probed by the query algorithm
on input $q$ that belong to the set $C_i(\bU)$. For any set of queries $Q' \subseteq Q$, we denote
the random variable $t_i(\bU, \bq)$ as the total number of cells probed from $C_i(\bU)$ when
processing a $\qop$ operation where the input $\bq$ is chosen uniformly at random
from $Q'$.

The first step of our proof will be to focus on individual epochs.

\begin{lemma}
\label{lem:weak_epoch}
Fix the random string $\cR$.
If $\beta = (wt_u)^2$, then for all epochs $i \in \{0,\ldots,k-1\}$,
$$\Pr\left[t_i(\bU, \bq_i) = \Omega\left(\frac{d'}{\log(t_u w)}\right)\right] \ge 1/2$$
where $\bq_i$ is chosen uniformly at random from $P_i$.
\end{lemma}

The proof of this lemma will use the cell sampling technique introduced by
Panigrahy {\em et al.}~\cite{PTW10} for static (non-oblivious) $\ANN$ lower bounds.
Their main idea is to, first,
assume the existence of an extremely efficient static
data structure that probes a small number of cells in expectation.
Next, they the show the existence of
a small subset of cells that {\em resolve} a very large
subset of possible queries where a query is resolved by a subset
of cells if the query does not probe any cells outside of the subset.
Afterwards, they show that the answers of the resolved queries
 reveal more bits of information about the input than the maximal amount
of information that can be stored about input in the subset of sampled cells.
This results in a contradiction showing there cannot exist
such an efficient static data structure that was original assumed.

However, to show that a lot of information is revealed by resolved queries,
the work of~\cite{PTW10} used several complex
combinatorial techniques.
These complex techniques end up being hard to scale
for the dynamic setting. To prove our dynamic lower bound, we first present
new ideas that simplify the static (non-oblivious) $\ANN$ proof using information theoretic arguments.
At a high level, we show that the set of resolved queries are a deterministic function of the sampled cells
which contain very little information about the inputs. This suffices to prove that
the resolved query set and inputs are almost independent. Since the input points are chosen
uniformly at random, it turns out that resolved queries will return a large number of input
points with constant probability, which would allow us to forego the complex techniques that
appear in~\cite{PTW10}.

Unfortunately, it turns out significantly larger problems appear when
moving to the dynamic setting even when using our simplifications.
Towards a contradiction, assume that
there exists an efficient data structure that probes $o(d'/\log(t_u w))$
cells from the set $C_i(\bU)$ in expectation. We can apply the cell
sampling technique to find a small subset $T_i \subset C_i(\bU)$ that resolves
a large number of queries. In this case, a query $q$ is resolved by $T_i$
if all cells that are probed by $q$ in the set $C_i(\bU)$ all belong to $T_i$.
Note, we do not project any restrictions on the cells probed by $q$ outside
the set $C_i(\bU)$. Once again, denote the set of queries resolved by $T_i$
using $Q_i(T_i)$. Using our information theoretic ideas, we want to show that
$Q_i(T_i)$ can be computed using only the little information stored
in the set of sampled cells $T_i$ and the client storage $M(\bU)$ as well as the
random string $\cR$.
In the dynamic case, the set of queries $Q_i(T_i)$ cannot be computed
using only $T_i$, $M(\bU)$ and $\cR$. As an example, consider a query
$q \in P_i$. During the execution of $q$, consider the first time a probe
is performed outside the set $T_i$. There is no way to determine
whether the probed cell exists in
$C_i(\bU)$ or not using only
the information in $T_i$, $M(\bU)$ and $\cR$.
As a result, it is impossible to accurately compute the set of
resolved queries $Q_i(T_i)$.

To get around this, we can attempt to also use $C_i(\bU)$ to compute
$Q_i(T_i)$. However, the set $C_i(\bU)$ is very large and may potentially
contain significantly
more information about $\bU_i$ compared to
the set of sampled cells $T_i$ and client storage $M(\bU)$.
As a result, we would not be able to prove our contradiction.
Instead, it turns out that computing $Q_i(T_i)$ only requires knowledge
of the addresses of $C_i(\bU)$. By the guarantees of statistical obliviousness,
we know that the addresses of $C_i(\bU)$ may not reveal too much
information about the underlying update operations $\bU$.
As a result, we can show that even though $C_i(\bU)$ is expressed using
many bits, that most of the bits cannot contain information about $\bU_i$.

One more issue that arises is that the above lemma is similar yet
crucially different than those used in lower bounds
for non-oblivious data structures.
In the standard application of the chronogram technique, 
the analogue of this lemma typically asserts that a single random query
$\bq_i$ must be \emph{simultaneously} 
``sensitive" to most epochs. That is, $\bq_i$
forces a large number of probes from cells in $C_i(\bU)$
for many (essentially all) epochs $i$ simultaneously.
Instead, our lemma says that for the all epochs,
there exists a special query distribution $\bq_i$ drawn uniformly at random from
$P_i$ built specially for that epoch $i$ that forces many probes
to cells in $C_i(\bU)$.
It turns out that this weaker lemma suffices for oblivious data structures.
We are able to use the fact that obliviousness must
hide the input query point from any adversary.
The main idea is that the adversary knows there exists some query point
from the set $P_i$ that must probe $\Omega(d'/\log(t_uw))$ cells from $C_i(\bU)$
to correctly answer the query. If the adversary views a query that probes significantly less cells
from $C_i(\bU)$, it can effectively deduce that the query does not come from the query set $P_i$ for
otherwise the answer of the query could not be correct. This observation by the adversary would
contradict obliviousness. As a result, we can essentially boost Lemma~\ref{lem:weak_epoch} into
the stronger variant below.

\begin{lemma}
\label{lem:strong_epoch}
If $\beta = (wt_u)^2$, there
exists a fixed query q such that
$$\E[t_i(\bU, q)] = \Omega\left(\frac{d'}{\log(t_uw)}\right)$$
for all epochs $i \in \{0,\ldots,k-1\}$.
\end{lemma}

The above lemma resembles the form of lemmata typically used in non-oblivious data structure lower bounds. We now show Lemma~\ref{lem:strong_epoch}
suffices to complete the lower bound by proving Theorem~\ref{thm:main}.

\begin{proof}[Proof of Theorem~\ref{thm:main}]
Note that sets of cells
$C_0(\bU), \ldots, C_{k-1}(\bU)$ are all disjoint and the random variable
$t_i(\bU, q)$ only counts the number of cells that are probed from $C_i(\bU)$.
Therefore, the total number of cells probed by
$t(\bU, q) = t_0(\bU, q) + \ldots + t_{k-1}(\bU, q)$.
There are $k = \Theta(\log_{\beta} (n/m))$ epochs.
Using linearity of expectation and Lemma~\ref{lem:strong_epoch}, it can be shown that
$$\E[t(\bU, q)] = \Omega(d' \log(n/m) / (\log(t_u w))^2).$$
The proof is completed by noting that $d = \Theta(d')$.
\end{proof}


\subsection{Bounding Cell Probes to Individual Epochs}

Towards
a contradiction, assume an extremely efficient data structure with
$t_i(\bU, \bq_i) = o(d'/\log(t_uw))$ where $\bq_i$ is drawn
uniformly at random from $P_i$. We apply the
cell sampling technique such that a small subset of cells $T_i \subset C_i(\bU)$
resolves a large number of queries $Q(T_i) \subseteq P_i$. 

\subsubsection{Cell Sampling}
\begin{lemma}
\label{lem:dynamic_sample}
Fix the random string $\cR$.
Suppose that $t_i(\bU, \bq_i) = o(\log\Phi/\log(t_u w)) = o(d'/\log(t_u w))$
where $\bq_i$ is drawn uniformly at random from $P_i$.
Then, there exists a subset of cells $T_i \subseteq C_i(\bU)$ with the
following properties:
\begin{itemize}
\item $|T_i| = \frac{n_i}{100w}$;
\item Let $Q_i(T_i)$ be all queries resolved by $T_i$ and probe
at most $2t_i(\bU, \bq_i)$ cells in $C_i(\bU)$.
Recall a query $q \in Q_i(T_i)$ is resolved by $T_i$ if every cell in $C_i(\bU)$
that is probed
when executing $q$ must exist in the subset $T_i$.
Then, $|Q_i(T_i)| \ge 2^{d'-1}/\Phi$.
\end{itemize}
\end{lemma}
\begin{proof}
For convenience, denote $t_i = t_i(\bU, \bq_i) = o(\log\Phi/\log(t_u w)) = o(d'/\log(t_u w))$.
Since we fixed the random string $\cR$, the randomness
of the data structure is strictly over the choice of updates from
the hard distribution $\bU$ and the random query $\bq_i$.
By Markov's inequality, there exists a subset of queries
$Q_i \subset P_i$ such that each $q \in Q_i$ probes at most
$2t_i$ cells in $C_i(\bU)$ and $Q_i$ contains at least $|P_i|/2 = 2^{d'-1}$
queries.

Consider the following random experiment where a subset
$\bT_i \subseteq C_i(\bU)$ is chosen uniformly at random from
all subsets with exactly $n/(100w)$ cells. Pick any query
$q \in Q_i$ probing at most $2t_i$ cells in $C_i(\bU)$.
We will analyze the probability that $q$ is resolved by $\bT_i$
over the random choice of $\bT_i$.
\begin{align*}
\frac{{|C_i(\bU)| - 2t_i \choose n_i/(100w) - 2t_i}}{{|C_i(\bU)| \choose n_i/(100w)}}
&\ge \frac{n_i/(100w) \cdot (n_i/(100w) - 1) \cdots (n_i/(100w) - 2t_i + 1)}{|C_i(\bU)| \cdot (|C_i(\bU)|-1) \cdots (|C_i(\bU)| - 2t_i + 1)}\\
&\ge \left(\frac{n_i/(100w) - 2t_i}{|C_i(\bU)|}\right)^{2t_i}\\
&\ge \left(\frac{n_i}{200|C_i(\bU)|w}\right)^{2t_i}\\
&\ge \left(\frac{1}{200t_uw}\right)^{2t_i}\\
&\ge \Phi^{-1}.
\end{align*}
The second last inequality uses the fact that $|C_i(\bU)| \le n_i t_u$ while
the last inequality uses the fact that $t_i = o(\log\Phi/\log(t_u w))$.
By linearity of expectation, we know that
$$\E[|Q_i(\bT_i)|] \ge |Q_i| \cdot \Phi^{-1} = 2^{d'-1}/\Phi.$$
As a result, there exists a subset $T_i \subset C_i(\bU)$
satisfying all the required properties.
\end{proof}

\subsubsection{Information from Resolved Queries}

Next, we will show that the resolved queries $Q_i(T_i)$ will report
a large number of points that are inserted by $\bU_i$.
Recall that a point in $\bU_i$ is reported by a query in $Q_i(T_i)$
if and only if it belongs to the neighborhood of $Q_i(T_i)$ denoted
by $\Gamma_r(Q_i(T_i)) \subseteq P_i$. Note, we only consider expansion
within the subcube $P_i$.
For convenience,
we fix $\bU_{-i}$, which consists of all updates outside of epoch $i$.

Towards a contradiction,
we will suppose that most points inserted by $\bU_i$ land
outside of $\Gamma_r(Q_i(T_i))$ and present an impossible
compression of $\bU_i$.
Formally, we construct a one-way encoding protocol from an encoder (Alice)
to a decoder (Bob). Alice receives as input $\bU$ and the random string
$\cR$. Bob will receive the addresses of cells in $C_i(\bU)$ denoted
by $C_i^{\addr}(\bU)$ and the random string $\cR$. The goal of Alice
is to encode the $n_i$ points inserted in $\bU_i$.
By Shannon's source coding theorem, the expected length of Alice's
encoding must be at least $H(\bU_i \mid C_i^{\addr}(\bU), \cR)$,
which we now analyze. 
In particular, we present an argument that the entropy of $\bU_i$ remains high
even conditioned on $C_i^{\addr}(\bU)$ due to statistical
obliviousness guarantees.
However, statistical obliviousness provides guarantees
using statistical distance which is not directly compatible with
our information theoretic arguments.
To do this, we present
the following lemma upper bounds the contributions of the positive
terms to the Kullback-Leibler
divergence between two distributions, in terms of their
statistical distance.
We note a similar lemma previously appeared in~\cite{BRW13}.

\begin{lemma} [Reverse Pinsker]    \label{lemma:reverse_pinsker}
Let $p(a,b)$ and $q(a,b)$ be two distributions over $A\times B$ in the
same probability space, and let
$S = \left\{ (a,b) : \lg \frac{p(a|b)}{q(a|b)} > 1 \right\}$.
Then, $p(S) < 2|p(a,b)-q(a,b)|$.
\end{lemma}

\begin{proof}
Let $\epsilon = | p(a,b) - q(a,b)| := 2 \max_T \{ p(T) - q(T) \} \geq 2(p(S) - q(S))$.
Rearranging sides, we have:

\begin{align*}
p(S) &\leq \epsilon/2 + q(S) \\
& < \epsilon/2 + (1/2) \sum_{(a,b) \in S} q(b)\cdot  p(a|b)\\
& \leq \epsilon/2 + (1/2) \sum_{(a,b) \in S} p(b) \cdot  p(a|b)
+ (1/2) \sum_{(a,b) \in S} |q(b)-p(b)| \cdot  p(a|b)\\
&\leq \epsilon/2 +  p(S)/2 + (1/2) \sum_{(a,b) \in S} |q(b)-p(b)|  \cdot p(a|b) \\
&\leq \epsilon/2 + p(S)/2  + (1/2) \sum_{b} |q(b) - p(b)|\\
& \leq \epsilon + p(S)/2
\end{align*}
where the second inequality follows from the fact for any $(a,b) \in S$,
$q(a|b) \geq p(a|b)/2$ by the choice of $S$.
\end{proof}

This lemma directly implies that $D_{KL}(p(a,b)||q(a,b)) \leq 2|p-q|_1 \cdot max_{a,b} \lg(p(a|b)/(q(a|b))) + 1$,
since the total contribution of terms outside $S$ is at most $\sum_{(a,b)} p(a|b) \leq 1$. Using the above, we show that
the entropy of $\bU_i$ conditioned on Bob's input remains large.

\begin{lemma}
\label{lem:entropy}
Consider any $\bU$ where all of $\bU_1,\ldots,\bU_{i-1}, \bU_{i+1},\ldots,\bU_{k-1}$ are fixed. That is, all update operations outside of epoch $i$ are fixed, 
and denote this fixed value by $\bU_{-i} = u_{-i}$. Then,
$$H(\bU_i \mid C_i^{\addr}(\bU), \cR, u_{-i}) = n_i \cdot \left(d' - o(1)\right).$$
\end{lemma}
\begin{proof}
We analyze the mutual information between $\bU_i$ and $\view_{\ds}(\bU)$.
Denote by $P$ and $Q$ the following distributions: $P \sim (\bU_i \mid \view_{\ds}(\bU), u_{-i})$ and
$Q \sim (\bU_i \mid u_{-i})$. By definition, 
\begin{align*}
	I(P; Q) &= \E_{v \sim \view_{\ds}(\bU)}\left[D_{KL}(P(\bU_i \mid v,u_{-i})\ \|\ Q(\bU_i \mid u_{-i}))\right]\\
	&\le 2 \cdot \E_v\left[\| P - Q \|_1 \right] 
	\cdot \max_{u_i, u_{-i}, v} \log \left(\frac{P(u_i \mid u_{-i}, v)}{Q(u_i)}\right) + 1\\
&= O\left( \frac{n_i \cdot d'}{\log^2 n}\right)
\end{align*}
where the first inequality is by Lemma \ref{lemma:reverse_pinsker}, and
the second is by the statistical-indistinguishability premise
that $\| P - Q \|_1 \le 1/\lg^2 n$, and the fact
that $\bU_i$ picks points uniformly at random and independent of $\bU_{-i}$. Hence the ratio between $P$ and $Q$ never exceeds $2^{n_i d}$.

Now, recall that $\cR$ is independent of $\bU_i$ and that $\bU_i$ is generated
independent of $\bU_{-i}$.
Therefore,
$$
H(\bU_i \mid C_i^{\addr}(\bU), \cR, u_{-i}) = H(\bU_i \mid C_i^{\addr}(\bU)).
$$
We can rewrite
$$
H(\bU_i \mid C^\addr_i(\bU)) = H(\bU_i) - I(\bU_i; C^\addr_i(\bU)) \ge (n_i \cdot d') \left(1 - O\left(\frac{1}{\log^2 n}\right)\right) \ge n_i \cdot (d' - o(1)).
$$
The second inequality uses the fact that
$C^{\addr}_i(\bU)$ appears in $\view_\ds(\bU)$. So,
$I(P; Q) = I(\bU_i; \view_\ds(\bU)) \ge I(\bU_i; C^{\addr}_i(\bU))$.
The last inequality uses the fact that $d' = \Omega(\log n)$.
\end{proof}

Going back to the original encoding protocol, we know that
Alice's expected encoding size must be at least
$H(\bU_i \mid C^{\addr}_i(\bU), \cR) = n_i \cdot (d' - o(1))$. We will utilize the fact that most points
inserted by $\bU_i$ land outside of $\Gamma_r(Q_i(T_i))$ to present
an impossible encoding scheme.

\begin{lemma}
\label{lem:dynamic_extract}
Fix $\bU_{-i}$, that is all update operations outside of epoch $i$.
With probability at least $1/2$ over the choice of $\bU_i$,
at least $n_i/8$ points in $\bU_i$ exist in the neighborhood of the set
of resolved queries, $\Gamma_r(Q_i(T_i))$. That is,
$$
\Pr[|\Gamma_r(Q_i(T_i)) \cap \bU_i| \ge n_i / 8] \ge 1/2.
$$
\end{lemma}
\begin{proof}
Towards a contradiction, suppose that the number of points in $\bU_i$
that land in $\Gamma_r(Q_i(T_i))$ is at least $n_i/8$ with
probability at most $1/2$.
Let $u_{-i}$ be the realization of $\bU_{-i}$.
We construct an impossible one-way communication protocol
for encoding $\bU_i$ which will contradict Shannon's source coding
theorem.

\paragraph{Alice's Encoding.} Alice receives as input
$\bU_i$, $u_{-i}$ and $\cR$.
\begin{enumerate}
\item Using $u_{-i}$, $\bU_i$ and $\cR$, execute
all operations to compute the cell sets $C_{k-1}(\bU), \ldots, C_0(\bU)$.
Afterwards, Alice finds the supposed $T_i$
of Lemma~\ref{lem:dynamic_sample}. To do this, Alice can iterate through
all subsets of $C_i(\bU)$ containing exactly $n_i/(100w)$ cells.
Alice can also compute query sets
$Q_i(T_i)$ and $\Gamma_r(Q_i(T_i))$.
Finally, Alice computes $F$ denoting the number of points
of $\bU_i$ in $\Gamma_r(Q_i(T_i))$.
\item If there are more than $n_i/8$ points of $\bU_i$ in
$\Gamma_r(Q_i(T_i))$, $F \ge n_i/8$, then Alice's encoding starts with
a $0$-bit. Alice encodes $\bU_i$ in the trivial manner using
$n_i \cdot d'$ bits.
\item Otherwise, suppose that less than $n_i/8$ points of $\bU_i$
land in $\Gamma_r(Q_i(T_i))$. That is, $F < n_i/8$.
In this case, Alice encodes
the contents and addresses of $T_i$ using $2w \cdot |T_i| = n/50$ bits.
Next, Alice encodes the set of cells last updated by operations after
epoch $i$. That is, the addresses and contents of cells in
$C_{i-1}(\bU),\ldots,C_0(\bU)$. The total
number of cells in these are
$n_i/\beta + n_i/\beta^2 + \ldots = \Theta(n_i/\beta)$ as $\beta \ge 2$.
Alice also encodes the client storage after executing all updates,
$M(u_{-i}, \bU_i)$ using $m = o(n)$ bits.
Alice encodes $F$ using $\log n_i$ bits and the indices of
$\bU_i$ whose points land in $\Gamma_r(Q_i(T_i))$ using
$\log {n_i \choose F}$ bits. Each of these $F$ points are
encoded trivially using $d'$ bits each. The remaining
$n_i - F$ points that land outside of $\Gamma_r(Q_i(T_i))$
are encoded using $\log(|P_i| - |\Gamma_r(Q_i(T_i))|)$
bits.
\end{enumerate}

\paragraph{Bob's Decoding.} Bob receives as input
$u_{-i}$, $C^{\addr}_i(\bU)$, $\cR$ and Alice's encoding.
\begin{enumerate}
\item If Alice's encoding starts with a $0$-bit, then Bob decodes
$\bU_i$ using the next $n_i \cdot d'$ bits in the trivial manner.
\item Otherwise, Bob executes all updates prior to epoch $i$
using $u_{-i}$ and $\cR$.
Bob decodes the addresses and contents of
$T_i \subset C_i(\bU)$ as well as the addresses and contents of
$C_{i-1}(\bU),\ldots,C_0(\bU)$.
At this point, Bob has the contents and addresses of all cell
sets $C_{k-1}(\bU),$ $\ldots,$ $C_{i+1}(\bU),$ $C_{i-1}(\bU),$ $\ldots,$ $C_0(\bU)$.
Additionally, Bob has the addresses of $C_i(\bU)$, $C^{\addr}_i(\bU)$,
but not the contents.
Using the next $m$ bits, Bob decodes the client storage $M(\bU)$
after executing all updates.
Bob attempts to execute each possible query in $P_i$ to compute
$Q_i(T_i)$. Note, Bob executes each query using $\cR$ and
$M(\bU_i)$ until the query attempts
to probe a cell with an address in $C^{\addr}_i(\bU) \setminus T^{\addr}_i$,
probes more than $2t_q$ cells or finishes executing.
As long as a query does not probe a cell in $C_i(\bU) \setminus T_i$,
Bob is able to accurately simulate the query.
As a result,
Bob accurately computes $Q_i(T_i)$ as well as $\Gamma_r(Q_i(T_i))$.
Next, Bob decodes $F$ as well as the $F$ indices of $\bU_i$ of points
that in $\Gamma_r(Q_i(T_i))$. For each of these $F$ points, Bob
uses the next $d'$ bits to decode them in the trivial manner. For the remaining
$n_i - F$ points, Bob decodes the point using the next $\log(|P_i| - |\Gamma_r(Q_i(T_i))|)$.
\end{enumerate}

\paragraph{Analysis.}
We start with the case of Alice's encoding is prepended with a $0$-bit.
For this scenario, Alice's encoding is always $1 + n_i \cdot d'$ bits.
Alice's encoding starts with a $0$-bit only in the case that
there are more than $n_i/8$ points of $\bU_i$ that land
in $\Gamma_r(Q_i(T_i))$ which happens with probability at most $1/2$
by our assumption towards a contradiction.

When Alice's encoding starts with a $1$-bit, Alice's encoding size in bits is at most
$$
1 + 2w(|T_i| + |C_{i-1}(\bU)| + \ldots + |C_0(\bU)|) + m + \log n_i + \log{n_i \choose F} + Fd' + (n_i - F)\log(|P_i| - |\Gamma_r(Q_i(T_i))|.
$$
By our choice of $\beta = (t_u w)^2$, we know that
$|C_{i-1}(\bU)| + \ldots + |C_0(\bU)| = \Theta(n_i/\beta)$.
By Lemma~\ref{lem:exp}, we know that $|\Gamma_r(Q_i(T_i))| \ge |Q_i(T_i)| \cdot \Phi \ge 2^{d'-1}$ as long as $|Q_i(T_i)|/\Phi \le 2^{d'-1}$.
If $|Q_i(T_i)|$ is too large, we can pick any arbitrary subset of size
$2^{d'-1}/\Phi$ and consider
the neighborhood of the subset.
As a result, we know that $\log(|P_i| - |\Gamma_r(Q_i(T_i))|) \le \log(2^{d'} - 2^{d'-1}) = d' - 1$.
Also, we note that $n_i \ge m^2$, so $m = o(n_i)$.
Note
that the encoding is maximized when $F = n_i/8$:
$$
n_i d' - \frac{7n_i}{8} + \frac{n_i}{50} + o(n_i) < n_i d' - \frac{n_i}{2} + o(n_i).
$$

Denote $p \le 1/2$ to be the probability that Alice's encoding starts with
a $0$-bit. Putting together the two cases, we get:
$$
p(1 + n_id') + (1-p)\left(n_id' - \frac{n_i}{2} + o(n_i)\right) < n_id' - \frac{n_i}{4} + o(n) < n_i (d' - o(1)) = H(\bU_i \mid C^{\addr}_i(\bU), \cR, u_{-i})
$$
since the encoding is maximized when $p = 1/2$. As a result, our
encoding is impossible as it contradicts
Shannon's source coding theorem.
\end{proof}


\subsubsection{Proof of Lemma~\ref{lem:weak_epoch}}

Lemma~\ref{lem:dynamic_extract} shows
that at least $n_i/8$ points of $\bU_i$ will land in the
set $\Gamma_r(Q_i(T_i))$ with high constant probability.
By Lemma~\ref{lem:dis}, all of these $n_i/8$ points will be reported
by at least one query $Q_i(T_i)$ with high probability.
We now show that the entropy contained in these $n_i/8$ points
is larger than the number of bits that may be stored in the contents of the
cells in $T_i$
to prove Lemma~\ref{lem:weak_epoch}.

\begin{proof}[Proof of Lemma~\ref{lem:weak_epoch}]
Towards a contradiction, suppose that $t_i(\bU,\bq_i) = o(d'/\log(t_u w))$.
Our assumption directly implies that
$\Pr[t_i(\bU,\bq_i) = \Omega(d'/\log(t_u w))] < 1/2$.
For convenience, fix all updates outside of epoch $i$
as $\bU_{-i} = u_{-i}$.
We present an impossible one-way communication protocol between
an encoder (Alice) and a decoder (Bob). Alice will attempt to encode
$\bU_i$ efficiently.
Both Alice and Bob will receive $\cR$. In addition, Bob will receive
the addresses of cells in $C_i(\bU)$ denoted by $C^{\addr}_i(\bU)$.
Alice's expected encoding size must at least $H(\bU_i \mid C^{\addr}_i(\bU),
\cR, u_{-i})$.
By Lemma~\ref{lem:entropy}, we know that
$H(\bU_i \mid C^{\addr}_i(\bU), \cR, u_{-i}) = n_i \cdot (d' - o(1))$.

\paragraph{Alice's Encoding.}
As input, Alice receives $\bU_i$, $u_{-i}$ and $\cR$.
\begin{enumerate}
\item Using $u_{-i}$, $\bU_i$ and $\cR$, execute
all operations to compute the cell sets $C_{k-1}(\bU), \ldots, C_0(\bU)$.
Afterwards, Alice finds the supposed $T_i$
of Lemma~\ref{lem:dynamic_sample}. To do this, Alice can iterate through
all subsets of $C_i(\bU)$ containing exactly $n_i/(100w)$ cells.
Alice can also compute query sets
$Q_i(T_i)$ and $\Gamma_r(Q_i(T_i))$.
Finally, Alice computes $F$ denoting the number of points
of $\bU_i$ in $\Gamma_r(Q_i(T_i))$.
\item If there are less than $n_i/8$ points in $\Gamma_r(Q_i(T_i))$
corresponding to $F < n_i/8$
or there exists two points in $\bU_i$ within distance at most
$0.49d'$, Alice's encoding will start with a $0$-bit. Alice will
encode $\bU_i$ in the trivial manner using $n_i \cdot d'$ bits.

\item Otherwise, Alice's encoding starts with a $1$-bit.
Alice encodes the addresses and contents of $T_i$ using $2w\cdot|T_i| = n/50$
bits.
Next, Alice encodes the addresses and contents of all cells
overwritten by an update operation after epoch $i$. That is, the cells
in $C_{i-1}(\bU),\ldots,C_0(\bU)$.
Alice also encodes the client storage after executing all update
operations, $M(\bU)$, using $m$ bits.
Using $n_i$ bits, Alice encodes whether each of the points in $\bU_i$
belong to $\Gamma_r(Q_i(T_i))$ or not.
For all $n - F$ points outside of $\Gamma_r(Q_i(T_i))$, Alice encodes
them using $d'$ bits each in the trivial manner. Afterwards, Alice
executes the queries in $Q_i(T_i)$ in some fixed order (such
as lexicographically increasing order). Each time a new point
in $\bU_i$ is reported, Alice encodes the index of the point in $\bU_i$
using $\log n_i$ bits completing the encoding.

\end{enumerate}

\paragraph{Bob's Decoding.}
Bob receives as input $u_{-i}$, $\cR$ and Alice's encoding.
\begin{enumerate}
\item If Alice's message starts with a $0$-bit, then Bob decodes
$\bU_i$ in the trivial manner using the next $n_i d'$ bits.
\item If Alice's encoding starts with a $1$-bit,
Bob executes all updates prior to epoch $i$
using $u_{-i}$ and $\cR$.
Bob decodes the addresses and contents of
$T_i \subset C_i(\bU)$ as well as the addresses and contents of
$C_{i-1}(\bU),\ldots,C_0(\bU)$.
At this point, Bob has the contents and addresses of all cell
sets $C_{k-1}(\bU),\ldots,C_{i+1}(\bU),C_{i-1}(\bU),\ldots,C_0(\bU)$.
Additionally, Bob has the addresses of $C_i(\bU)$, $C^{\addr}_i(\bU)$,
but not the contents.
Using the next $m$ bits, Bob decodes the client storage $M(u_{-i}, \bU_i)$
after executing all updates.
Bob attempts to execute each possible query in $P_i$ to compute
$Q_i(T_i)$.
Note, Bob executes each query using $\cR$ and
$M(\bU_i)$ until the query attempts
to probe a cell with an address in $C^{\addr}_i(\bU) \setminus T^{\addr}_i$,
probes more than $2t_q$ cells or finishes executing.
As long as a query does not probe a cell in $C_i(\bU) \setminus T_i$,
Bob is able to accurately simulate the query.
As a result,
Bob accurately computes $Q_i(T_i)$ as well as $\Gamma_r(Q_i(T_i))$.
Using the next $n_i$ bits, Bob decodes whether each point in $\bU_i$
belongs to $\Gamma_r(Q_i(T_i))$.
Bob decodes all $n_i - F$ points outside of $\Gamma_r(Q_i(T_i))$ using
the next $(n_i - F) \cdot d'$ bits in the trivial manner.
To decode the $F$ points in $\Gamma_r(Q_i(T_i))$, Bob will execute
the queries in $Q_i(T_i)$ in the same fixed order as Alice.
Each time a new point is reported, Bob uses the next $\log n_i$ bits
to decode the point's index in $\bU_i$ completing the decoding procedure.
As Alice's encoding starts with a $1$-bit only when
no two points in $\bU_i$ are within distance $0.49d'$, we know that
all points in $\Gamma_r(Q_i(T_i))$ will be reported by at least one
query in $Q_i(T_i)$.
\end{enumerate}

\paragraph{Analysis.}
We now analyze the expected length of Alice's encoding. If Alice's
encoding starts with a $0$-bit, we know that Alice's encoding is
exactly $1 + n_id'$ bits. Alice's encoding starts with a $0$-bit
only when less than $n_i/8$ points land in $\Gamma_r(Q_i(T_i))$ or
there exists two points in $\bU_i$ within distance at most $0.49d'$.
This occurs with probability at most $1/2 + 1/n$ by
Lemma~\ref{lem:dynamic_extract} and Lemma~\ref{lem:dis}.

On the other hand, consider the case when Alice's encoding starts
with a $1$-bit. In this case, Alice's encoding length is
$$
1 + 2w(|T_i| + |C_{i-1}(\bU)| + \ldots + |C_0(\bU)|) + m + n_i + (n_i - F)d' + F\log n_i.
$$
Note, that $|C_{i-1}(\bU)| + \ldots + |C_0(\bU)| = \Theta(n_i/\beta)$ by
our choice of $\beta = (t_uw)^2$. We chose epochs such that $n_i \ge m^2$,
so $m = o(n_i)$. For sufficiently large $d' > \log n_i$, we know that
Alice's encoding is maximized when $F = n_i/8$:
$$
\frac{7n_id'}{8} + \frac{n_i \log n_i}{8} + o(n_i\log n_i).
$$

Denote $p \le 1/2 + 1/n$ as the probability that Alice's encoding starts with a $0$-bit.
Then, Alice's encoding length in expectation is
$$
p(1 + n_id') + (1-p)\left(\frac{7n_i d'}{8} + \frac{n_i \log n_i}{8} + o(n_i \log n_i)\right)
$$
which is maximized when $p = 1/2 + 1/n$. So, Alice's expected encoding length
is at most
$$
\frac{15n_i}{16} d' + \frac{n_i}{16} \log n_i + o(n_i \log n_i) < n_i (d' - o(1)) = H(\bU_i \mid C^{\addr}_i(\bU), \cR, u_{-i}).
$$
This contradicts Shannon's source coding theorem completing the proof.
\end{proof}

\subsection{Bounding Cell Probes to All Epochs}
In this section, we complete the proof of
Lemma~\ref{lem:strong_epoch} using Lemma~\ref{lem:weak_epoch}.
Our proof
for Lemma~\ref{lem:strong_epoch} will apply obliviousness to Lemma~\ref{lem:weak_epoch}.
The main idea is that any adversary with the ability
can view the number of probes to cells
in the sets $C_{k-1}(\bU), \ldots, C_1(\bU)$.
Lemma~\ref{lem:weak_epoch} states that the expected running time for
queries chosen uniformly at random from $P_i$
requires probing $\Omega(d'/\log(t_u w))$ cells from
$C_i(\bU)$ with high constant probability.
If queries from outside of the $P_i$ were to probe significantly
less cells from $C_i(\bU)$, the adversary can distinguish queries that lie in
$P_i$ as opposed to those outside which would contradict
the obliviousness of the data structure. We now formalize these ideas to prove Lemma~\ref{lem:strong_epoch}.

\subsubsection{Proof of Lemma~\ref{lem:strong_epoch}}

\begin{proof}[Proof of Lemma~\ref{lem:strong_epoch}]
By Lemma~\ref{lem:weak_epoch}, we know that if we pick
a query point $\bq_i$ uniformly at random from a subset of $P_i$,
then $\Pr[t_i(\bU, \bq_i) \ge \gamma (d'/\log(t_u w))] \ge 1/2$ for some constant $\gamma$ and for all
epochs $i \in \{0,\ldots,k-1\}$.
We now consider two sequences of operations which both start with $\bU$. For the first sequence,
the query $\bq_i$ is chosen uniformly at random from $P_i$. For the second sequence,
the query is chosen as any query point outside of $P_i$ (that is $q \notin P_i$).

We note that an adversary can compute the sets of cells
$C_{k - 1}(\bU), \ldots, C_0(\bU)$ by simply executing the
update operations in $\bU$ and keep tracking of the last time the contents of a cell were updated.
Furthermore, for any query point $q$, an adversary can compute
$t_i(\bU, q)$ by simply counting the number of probes performed to cells in the set $C_i(\bU)$.

Suppose that
$t_i(\bU, q) < (\gamma / 4) \cdot (d'/\log(t_u w))$ which implies that
$\Pr[t_i(\bU, q) \ge \gamma (d'/\log(t_uw))] \le 1/4$
by Markov's inequality. The adversary
can now apply the following distinguisher to differentiate the two sequences where the final
query is $\bq_i$ chosen uniformly at random from $Q_i$ or $q \notin Q_i$.
The adversary computes the number of probes to cells in the set $C_i(\bU)$. If the number
of probes is less than $\gamma \cdot (d'/\log(t_uw))$, then the adversary outputs $0$.
Otherwise, the adversary outputs $1$. As a result, the adversary distinguishes the two sequences
with probability at least $1/4$ contradicting obliviousness. If we pick the query point $q$
such that $q \notin P_i$ for all epochs $i \in \{k-1,\ldots,0\}$,
we apply the above result simultaneously
to all epochs.

Note each epoch must contain at least $\min\{\sqrt{n}, m^2\}$. Furthermore,
epochs grow geometrically by a $\beta = (t_uw)^2$ factor and there
are $\Theta(n)$ update operations in total.
So, there are $k = \Theta(\log(n/m)/\log(t_uw))$ epochs completing the proof.
\end{proof}

\section{Oblivious Dynamization}
Let $\calP$ be a decomposable problem and suppose that we have an oblivious static data structure
that solves $\calP$ that holds $n$ items which requires storage of $S(n)$ cells, preprocessing of
at most $P(n)$ cell probes before queries and answers queries in amortized $Q(n)$ cell probes.
The static data structure has two functions: $\init^\st$ and $\qop^\st$.
The preprocessing function, $\init^\st$ takes as input an encryption key
$\cK^\enc$ and a set of items encrypted under $\cK^\enc$.
The output of the preprocessing function is the data structure's memory
as well as a query key $\cK^\st$. The query algorithm
takes as input the query key
$\cK^\st$ as well as the queried argument and outputs the result
as well as the possibly updated static data structure's memory.
We assume that both the preprocessing and queries are
performed obliviously. That is, the adversary's view of the preprocessing
and queries are independent of the underlying items and sequence of operations.
Using this oblivious static data structure in a blackbox manner,
we will construct an oblivious dynamic data structure which support
updating the underlying data.

\begin{theorem}
If there exists an oblivious static data structure for a decomposable problem $\calP$ of $n$ items
        with storage of $S^\st(n)$ cells, preprocessing of $P^\st(n)$ cell probes and amortized $Q^\st(n)$ cell
probes for queries,
then there exists an oblivious dynamic data structure for $\calP$ using
$S^\dy(n) = O(\sum_{i = 1}^{\log n} S^\st(2^i))$ cells of storage and
        amortized $Q^\dy(n) = O(\sum_{i=1}^{\log n} Q^\st(2^i) + \sum_{i=1}^{\log n} \frac{P^\st(2^i)} {2^i})$ cell probes for each query/insert operation.
\end{theorem}
\begin{proof}
We assume that the oblivious dynamic data structure is initially empty
and assume that the number of operations, $n$, is a power of two
for convenience. When $n$ is not a power of two, one can replace all
$\log n$ with $\ceil{\log n}$ to get the correct bounds.
We construct our dynamic data structure by initializing
$\log n$ levels of geometrically increasing levels.
Level $i$
will be initialized using an oblivious static data structure
with a data set of $2^i$ items.
We will denote level $i$ as $L_i$.
To satisfy the requirements of obliviousness,
we must hide from the adversary whether we are performing a
query or insertion operation.
To do this, we simply perform both for each operation.
In particular, we will perform a query first before performing the
insertion operation where exactly one of
the query or the insertion will be a fake operation.
Fake insertions will insert a $\perp$ and not affect future operations
while fake queries will perform an arbitrary query and ignore the result.
We now formally present our $\init^\dy$, $\qop^\dy$ and $\uop^\dy$ algorithms.

\paragraph{$(\cK^\enc, \cK^\st_1,\ldots,\cK^\st_{\log n}), (L_1,\ldots,L_{\log n}, S_1,\ldots,S_{\log n}) \leftarrow \init^\dy$().} Our preprocessing algorithm simply initializes all $\log n$ levels to be empty and generate an encryption key.
	\begin{enumerate}
		\item Generate encryption key $\cK^\enc$.
		\item For each $i = 1,\ldots, \log n$:
			\begin{enumerate}
				\item Set $L_i \leftarrow \perp$.
				\item Set $\cK^\st_i \leftarrow \perp$.
				\item Set $S_i \leftarrow \emptyset$.
			\end{enumerate}
	\end{enumerate}

	\paragraph{$r, (L_1,\ldots,L_{\log n}) \leftarrow \qop^\dy((q, \cK^\st_1,\ldots,\cK^\st_{\log n}), (L_1,\ldots,L_{\log n}))$.} Our query algorithm receives as input a query $q$.
\begin{enumerate}
\item Set $r \leftarrow \perp$.
\item For each $i = 1, \ldots, \log n$:
\begin{enumerate}
\item If $L_i \ne \perp$, then execute $r_i \leftarrow \qop^\st(\cK^\st_i, q, L_i)$.
\end{enumerate}
\item Return $r$.
\end{enumerate}

	\paragraph{$(\cK^\st_1,\ldots,\cK^\st_{\log n}), (L_1,\ldots,L_{\log n}, S_1,\ldots,S_{\log n}) \leftarrow \uop^\st(\cK^\enc, x)$.} Our insertion algorithm receives as input the item that should be inserted $x$.
\begin{enumerate}
\item Find minimum $k$ such that $L_i = \perp$.
\item Set $S_k \leftarrow \Enc(\cK^\enc, \{x\} \cup S_1 \cup \ldots \cup S_{k-1})$.
\item Set $(L_k, \cK^\st_k) \leftarrow \init^\st(\cK^\enc, S_k)$
	where $\cK^\st_k$ is the privacy key used to query the oblivious static
		data structure.
\item For each $i = 1, \ldots, k - 1$:
	\begin{enumerate}
		\item Set $L_i \leftarrow \perp$.
		\item Set $\cK^\st_i \leftarrow \perp$.
		\item Set $S_i \leftarrow \emptyset$.
	\end{enumerate}
\end{enumerate}

We now analyze the costs for our dynamic data structure. Note that the $\qop$ algorithm
requires performing at most $\log n$ queries to static data structures of size at most $n$
	resulting in $\sum_{i=1}^{\log n} Q^\st(2^i) = O(\log n \cdot Q^\st(n))$ cell probes. For the $\uop$ algorithm, we perform an amortized
analysis over $n$ queries. Level $i$ is reconstructed every $n/2^i$ operations with $P^\st(2^i)$ cell probes.
As a result, the total cost over all $n$ queries is $\sum_{i=1}^{\log n} P^\st(2^i) \cdot n / 2^i$ or
$\sum_{i = 1}^{\log n} P^\st(2^i) / 2^i$ amortized over all $n$ queries. The total storage is always at most
$\sum_{i=1}^{\log n}S^\st(2^i)$.

Finally, we analyze the obliviousness of our data structure. We note that the schedule of constructing
static data structures as well as querying is completely deterministic and independent of the stored data,
input arguments to operations as well as previous updates. As we assume the queries to each static data structure
and the preprocessing to construct the static data structure are oblivious, our dynamic data structure also
provides obliviousness.
\end{proof}

\small
\bibliographystyle{alpha}
\bibliography{biblio}

\end{document}